\documentclass[11pt]{article}

\usepackage{amsthm, amsbsy, amssymb, amsfonts}
\usepackage{fullpage,times}

\newtheorem{theorem}{Theorem}
\newtheorem{lemma}[theorem]{Lemma}

\newtheorem{claim}[theorem]{Claim}

\newtheorem{remark}[theorem]{Remark}
\newtheorem{definition}[theorem]{Definition}
\newtheorem{proposition}[theorem]{Proposition}

\newcommand{\ceil}[1]{\left\lceil #1 \right\rceil}
\newcommand{\floor}[1]{\left\lfloor #1 \right\rfloor}
\def\E{\mathop{\rm E}}
\newcommand{\defeq}{\stackrel{\Delta}{=}}
\newcommand{\cN}{{\mathcal N}}
\newcommand{\dist}{\mathsf{dist}}
\newcommand{\adv}{\mathsf{adv}}
\newcommand{\reals}{\mathbb{R}}
\newcommand{\cT}{\mathcal{T}}
\newcommand{\rv}{\mathbf{v}}
\newcommand{\cB}{\mathcal{B}}
\newcommand{\cS}{\mathcal{S}}
\newcommand{\cA}{\mathcal{A}}

\newcommand{\hT}{\hat{\mathcal{\cT}}}
\newcommand{\zbar}{\overline{z}}
\newcommand{\xbar}{\overline{x}}
\newcommand{\hv}{\hat{v}}
\newcommand{\hg}{\hat{g}}
\newcommand{\hb}{\hat{b}}
\newcommand{\sign}{\mathsf{sign}}
\newcommand{\depth}{\mathsf{depth}}
\newcommand{\gnd}{\mathsf{gnd}}
\newcommand{\xnd}{\mathsf{xnd}}

\begin{document}

\title{How Hard is Computing Parity with Noisy Communications?\footnote{A preliminary version of this work appeared 
in the Proceedings of the
Nineteenth Annual ACM-SIAM Symposium on Discrete Algorithms, 2008, pp. 1056-1065.}}

\author{Chinmoy Dutta\thanks{Twitter Inc., San Francisco, USA. 
email: {\tt chinmoy@twitter.com}} \\
\and
Yashodhan Kanoria\thanks{Columbia Business School, New York, USA. 
email: {\tt ykanoria@columbia.edu} The work was done while this author was at
Indian Institute of Technology, Mumbai, INDIA.} \\
\and
D. Manjunath\thanks{Indian Institute of Technology, Mumbai, INDIA. 
email: {\tt dmanju@ee.iitb.ac.in}} \\
\and
Jaikumar Radhakrishnan\thanks{Tata Institute of Fundamental Research, Mumbai, INDIA. 
email: {\tt jaikumar@tifr.res.in}}}

\date{}
\maketitle

\begin{abstract} 
\small\baselineskip=9pt
We show a tight lower bound of $\Omega(N \log\log N)$ on the number of
transmissions required to compute the parity of $N$ input bits with constant
error in a noisy communication network of $N$ randomly placed sensors, 
each having one input bit and communicating with others using
local transmissions with power near the connectivity threshold. 
This result settles the lower bound question left open by Ying, Srikant and
Dullerud (WiOpt 06), who showed how the sum of all the $N$ bits can be computed using
$O(N \log\log N)$ transmissions. The same lower bound has been shown to hold for a host of
other functions including majority by Dutta and Radhakrishnan (FOCS 2008).

Most works on lower bounds for communication networks considered mostly the
full broadcast model without using the fact that the communication in real networks 
is local, determined by the power of the transmitters.  In fact, in full broadcast 
networks computing parity needs $\theta(N)$ transmissions. To obtain our lower bound 
we employ techniques developed by Goyal, Kindler and Saks (FOCS 05), who showed 
lower bounds in the full broadcast model by reducing the problem to a model of 
noisy decision trees. However, in order to capture the limited range of transmissions
in real sensor networks, we adapt their definition of noisy decision trees and allow
each node of the tree access to only a limited part of the input. 
Our lower bound is obtained by exploiting special properties
of parity computations in such noisy decision trees.
 
\end{abstract}

\section{Introduction}
\label{sec:intro}

Since inexpensive wireless technology and sensing hardware have become
widely available and are heavily used, much recent effort has been
devoted to developing models for these networks and protocols based on
these models. A wireless sensor network consists of sensors that
collect and cooperatively process data in order to compute some global
function. The sensors interact with each other by transmitting
wireless messages based on some protocol. The protocol is required to
tolerate errors in transmissions since wireless messages typically are 
noisy.

In the problem we study, each sensor is required to detect a bit;
then, all the sensors are required to collectively compute the parity
of these bits. The difficulty of this task, of course, depends on the
noise and the connectivity of the network. In this paper, we assume
that each bit sent is flipped (independently for each receiver) with
probability $\epsilon>0$ during transmission. As for connectivity, we
adopt the widely used model of random planar networks. Here the
sensors are placed randomly and uniformly in a unit square. Then each
transmission is assumed to be received (with noise) by the sensors
that are within some prescribed radius of the sender. The radius is
determined by the amount of power used by the sensors, and naturally
one wishes to keep the power used as low as possible, perhaps just
enough to ensure that the entire network is connected. If the network 
is not connected then it cannot be expected to compute a function like 
parity which depends on all the input bits. It has been
shown by Gupta and Kumar~\cite{Gupta00} that the threshold radius for
connectivity is $\theta\left(\sqrt{\frac{\ln N}{N}}\right)$ for a random
planar network of $N$ sensors placed in a unit square. With a
radius much smaller than this the network will not be connected almost
surely, and with radius much larger it will be connected almost
surely.

Our work is motivated by a protocol presented by Ying, Srikant and
Dullerud~\cite{Ying06} for computing the sum of all the bits (and hence any
symmetric functions of these bits).  They showed that even with
radius of transmission just near the connectivity threshold, and
constant noise probability, one can compute the sum using a total of
$O(N \log \log N)$ transmissions. They observed the (trivial) lower
bound of $N$ transmissions (for every sensor must send at least
one message), but left open the possibility of better upper bounds.
One can compute the parity of the input bits from their sum; in fact, Ying et al.
suggested that parity computation might be significantly easier than computing the sum.  
In this work, we prove a lower bound showing that the protocol of Ying et al.
is optimal up to constant factors for computing the parity (and hence, also the sum)
of the input bits. In order to state our result formally we need to define the model of noisy 
communication networks.
 
\begin{definition}[Noisy communication network and protocol] \hspace{0.2in}
\label{def:noisy-network-protocol}
A communication network is an undirected graph $G$ whose vertices
correspond to sensors and edges correspond to communication links.
A message sent by a sensor is received by all its neighbors.
\begin{description}
\item[Noise:] In an $\epsilon$-noise network, the messages are
subjected to noise as follows. Suppose sensor $v$ sends bit $b$ in
time step $t$. Each neighbor of $v$ then receives an independent
noisy version of $b$; that is, the neighbor $w$ of $v$ receives the
bit $b \oplus \eta_{w,t}$, where $\eta_{w,t}$ is an $\epsilon$-noisy
bit (that takes the value $1$ with probability $\epsilon$ and $0$
with probability $1-\epsilon$), these noisy bits being mutually
independent for different neighbors.

\item[Input:] An input to the network is an assignment of bits to the sensors, 
and is formally an element of $\{0, 1\}^{V(G)}$.
 
\item[Protocol:]  A protocol on $G$ for computing a function $f:\{0,1\}^{V(G)}
\rightarrow \{0,1\}$ works as follows. The sensors take turns to send single 
bit messages, which are received only by the neighbors of the sender.
In the end, a designated sensor $v^* \in V(G)$ declares the
answer. The cost of the protocol is the total number of bits
transmitted.  A message sent by a sensor in some time step is a function of the bits
that it possesses, which include its input bit and the noisy copy 
of the bits transmitted by its neighbors until then.  The protocol with cost $T$ is thus
specified by a sequence of $T$ vertices $\langle v_1,v_2,\ldots,v_T\rangle$ and a 
sequence of $T$ functions $\langle g_1,g_2,\ldots,g_T \rangle$, where 
$g_t:\{0,1\}^{j_t} \rightarrow \{0,1\}$ and $j_t$ is the number of bits possessed by 
$v_t$ before time step $t$. Furthermore, $v_T=v^*$, and the final answer is obtained by computing $g_T$. 
Note that in our model the number of transmissions is the same for all
inputs.

\item[Error:] Such a protocol is said to be a $\delta$-error protocol,
if for all inputs $x \in \{0,1\}^{V(G)}$, $\Pr[\mathsf{output} = f(x)] \geq 1-\delta$.
Here the probability is over the noise in the communication channel as well as the
internal randomness, if any, of the protocol.
\end{description}
\end{definition}

In this paper, we consider networks that arise out of random placement 
of sensors in the unit square.

\begin{definition}[Random planar network] \hspace{0.2in}
\label{def:random-planar-network}
A random planar network $\cN(N,R)$ is a random variable whose values
are undirected graphs. The distribution of the random variable depends
on two parameters: $N$, the number of vertices, and $R$, the
transmission radius.  The vertex set of $\cN(N,R)$ is $V(\cN) = \{P_1, P_2,
\ldots , P_N\}$. The edges are determined as follows. First, these
vertices are independently placed at random, uniformly in the unit
square $[0,1]^2$. Then,
\[ E(\cN) = \{ (P_i,P_j):  \dist(P_i,P_j) < R\},\]
where $\dist(P_i,P_j)$ is the Euclidean distance between vertices $P_i$ and $P_j$.
\end{definition}

The result in this paper is the following.

\begin{theorem}[Lower bound for parity] \hspace{0.2in}
\label{thm:lb-parity}
Let $R \leq N^{-\beta}$ for some $\beta >0$.  
Let $\delta < \frac{1}{2}$ and $\epsilon \in (0,1)$. Then, 
with probability $1-o(1)$ over the random variable $\cN(N,R)$, every 
$\delta$-error protocol on $\cN(N,R)$ with $\epsilon$-noise for computing 
the parity function $\oplus: \{0, 1\}^{V(\cN)} \rightarrow \{+1,-1\}$ 
requires $\Omega(N \log \log N)$ transmissions.
\end{theorem}

\begin{remark}
Our definition of noise assumes that all transmissions are subjected
to noise with probability exactly $\epsilon$. In the literature, other 
models of error have been considered. Some protocols work even
in the weaker model where this probability is at most $\epsilon$.  Our
lower bound holds for the stronger model with the noise parameter being exactly $\epsilon$, 
and hence is also applicable to the weaker model. 
\end{remark}

\begin{remark}
We require only an upper bound on the transmission radius.
However, the result is meaningful only when $R =\Omega(\sqrt{\frac{\log N}{N}})$, 
for otherwise, with high probability, the network is not connected and cannot be 
expected to compute any function that depends on all its input bits.
\end{remark}

\begin{remark}
Trivially, this lower bound also holds for computing the sum of the input bits. 
\end{remark}

\subsection{Related work}
The most commonly studied noisy communication model allows full
broadcasts, that is, all sensors receive all messages (with
independent noise). In this model, Gallager~\cite{Gallager88} considered the
problem of collecting all the bits at one sensor, and showed how this
could be done using $O(N \log \log N)$ transmissions; this implies the
same upper bound for computing any function of the input bits. More
recently, in a remarkable result, Goyal, Kindler and Saks~\cite{Goyal05}
showed that Gallager's protocol was the best possible for collecting
all the bits.  However, they do not present any boolean function for
which $\Omega(N \log \log N)$ transmissions are required.

In the full broadcast model, protocols for computing specific
functions have also been studied in the literature. Feige and
Raghavan~\cite{Feige00} presented a protocol with $O(N \log ^* N)$
transmissions for computing the OR of $N$ bits; this result was
improved by Newman~\cite{Newman04}, who gave a protocol with $O(N)$
transmissions. For computing threshold functions Kushilevitz and
Mansour~\cite{Kushilevitz98} showed a protocol with $O(N)$ transmissions, assuming
that all messages are subject to noise with probability exactly
$\epsilon$. Under the same assumption, Goyal, Kindler and Saks~\cite{Goyal05}
showed that the sum of all the bits (and hence all symmetric
functions) could be computed with $O(N)$ transmissions.

In this paper we are concerned with networks arising from random
placement of sensors, where considerations of power impose stringent
limits on the transmission radius. In this model, Ying,
Srikant and Dullerud~\cite{Ying06} presented a protocol for computing the sum of all the bits as
mentioned above. Kanoria and Manjunath~\cite{Kanoria07} gave a protocol that uses $O(N)$
transmissions to compute the OR function. However, no non-trivial
lower bound that apply specifically to communication networks with limited
transmission radius had appeared in the literature before this work. Subsequent to the 
initial presentation of this work~\cite{Dutta08a}, Dutta and Radhakrishnan~\cite{Dutta08b} showed that 
the same lower bound of $\Omega(N \log \log N)$ holds for computing a host of boolean functions 
including the majority function.

\subsection{Techniques}

We now present an overview of the proof technique used to derive our
lower bound.  As we explain in more detail in the Section~\ref{sec:lb-parity}, the
proof has two parts. The first part is geometric. Since the transmission 
radius is limited, it is possible to decompose the nodes of the communication
network into clusters. The nodes in the interior of each cluster will
continue to receive inputs and will be called {\it input nodes}, but those on the 
boundary will have their inputs fixed (arbitrarily) and thereby become 
{\it auxiliary nodes} that still participate in the protocol by sending and receiving
messages. This decomposition of the communication network 
into clusters ensures that any node can receive transmissions from input nodes of at most
one cluster. This allows us to view the protocol as a combination of several subprotocols 
acting on different clusters and interacting with each other via the auxiliary nodes.
This graph theoretic decomposition is based on routine
arguments involving the distribution points chosen independently and uniformly at random 
on the unit square. 

The second part of the proof is combinatorial and concerns arguing that the subprotocols
acting on different clusters of the decomposed network can be assumed to be independent of each other.
This part is not straightforward and we need to revisit the arguments
used by Goyal, Kindler and Saks~\cite{Goyal05} to obtain their lower bounds. 
A key insight in their proof was that protocols in noisy communication networks
could be translated into what they called {\it Generalized Noisy Decision trees} ($\gnd$ trees). 
We adapt their argument to our setting. For us it is important
to ensure that the decomposition of the network (which was the
consequence of the limited transmission radius) is reflected in the
noisy decision trees we construct. So, we define a notion of noisy
decision trees appropriate for our setting, where we allow each node of the tree
access to the inputs of only one cluster. We show how efficient
protocols on decomposed networks can be translated to such decision
trees of small depth. 

The argument this far was general and did not use the fact that the
ultimate goal of the protocol is to compute the parity function. 
Next we show that we can rearrange the decision tree so that
the queries made to the variables in the same cluster of the decomposition
appear at adjacent levels of the tree. This part crucially depends on
the fact that we are trying to compute the parity function. After the
rearrangement, we can view the entire computation as a sequence of
noisy decision tree computations, one for each cluster.  We conclude
that in order to have low overall error, the computation in each
cluster must have vanishingly small error probability.  At this stage
we can directly apply a result of Goyal, Kindler and Saks~\cite{Goyal05}, which
states that any decision tree that computes the parity function with
error $o(1)$ must have superlinear depth.  This dependence of depth on
error is strong enough to yield our lower bound. 

The interesting feature of this argument is that we work with
appropriately defined decision trees instead of directly with the
decomposed protocol.  Once inputs of processors have been fixed, they
become auxiliary. However, they continue to participate in the
protocol. In particular, they receive transmissions from processors
with inputs and can potentially aid error correction by providing
additional reception diversity, which is crucially exploited in many
of the upper bounds. So it is not true that our decomposition
immediately breaks the protocol into independent subprotocols,
operating separately on different clusters. Nevertheless, when we
translate the decomposed protocol into our model of decision trees, we
can view the computation of the entire decision tree as a combination
of independent decision subtrees, operating separately on different
clusters. This provides us the required product property, from which
one easily deduces that each individual subtree must compute the
parity within its cluster very accurately. For an detailed discussion 
of this technique as well as those developed to analyze functions 
where we do not have the product property, we refer the reader to the 
Phd thesis~\cite{Dutta09}.

\subsection{Organization of the paper}
Section~\ref{sec:prelims} presents some definitions and notations. 
In Section~\ref{sec:lb_parity}, we state two lemmas corresponding to the two
parts of the argument, and derive the lower bound for parity.  
The details of the first part of the argument are presented in 
Secction~\ref{sec:network_decomposition}. The second part of the argument is spread
over Sections~\ref{sec:protocol_to_read_once_tree} and~\ref{sec:read_once_tree_advantage}. 
We conclude the paper in Section~\ref{sec:conclusions}.
\section{Preliminaries}
\label{sec:prelims}

In our proof, some of the nodes in the network will receive no input. 
We now introduce the terminology applicable in such situations.

\begin{definition} [Input and auxiliary nodes] \hspace{0.2in}
Let $G = (V,E)$ be a communication network. We partition the set of nodes, $V$, into 
the set of input nodes, $I$, and the set of auxiliary nodes, $A$. Nodes in $I$ receive inputs
and those in $A$ do not receive any input but have their input bits fixed arbitrarily.
An input to such a network is an element of $\{0,1\}^I$ and a protocol on such a network 
computes a function $f: \{0,1\}^I \rightarrow \{0,1\}$.
\end{definition}

Next we formalize the notions of network decomposition and bounded protocols on such decomposed networks.

\begin{definition}[Network decomposition and bounded protocols] \hspace{0.2in}
\label{def:network-decomposition}
Let $G=(I \cup A, E)$ be a communication network. An $(n,k)$-decomposition
of $G$ is a partition of the set of nodes of $G$ of the form $I = I_1 \cup 
\cdots \cup I_k$ and $A = A_0 \cup A_1 \cup \cdots \cup A_k$ such that for $j=1,\ldots,k$,
\begin{enumerate}
\item[(P1)] $|I_j| = n$, and
\item[(P2)] the neighborhood of $I_j$ is contained in $I_j \cup A_j$.
\end{enumerate}
A protocol $\Pi$ on $G$ is said be a $(d,D)$-bounded
protocol with respect to the decomposition $\langle A_0,(I_j,A_j): j=1,\ldots,k\rangle$ if
for $j = 1, \ldots, k$,
\begin{enumerate}
\item[(P3)] a node in $I_j$ makes at most $d$ transmissions, and
\item[(P4)] all nodes in $I_j \cup A_j$ put together make at most $D$ transmissions.
\end{enumerate}
We use the notation $\epsilon$-noise $(n,k,d,D)$-protocol to mean a $(d,D)$-bounded
protocol for some $(n,k)$-decomposed network with noise parameter $\epsilon$.
\end{definition}

As stated earlier, we will use the method of Goyal, Kindler and Saks~\cite{Goyal05} 
to translate a communication protocol into a noisy decision tree. 
We now present the terminology for noisy decision trees.

\begin{definition}[Decision tree]  \hspace{0.2in}
\label{def:decision-tree}
Let $S$ be an arbitrary set and $k$ be a positive integer. 
A decision tree $\cT$ for the set of inputs $S^k$ is a balanced tree where
each internal node $v$ is labelled by a pair $\langle i_v, g_v
\rangle$ where $i_v \in [k]$, $g_v: S \rightarrow C_v$, and
$C_v$ is the set of children of $v$. We call the tree to be a {\bf noisy} decision tree
if the functions $g_v$ are noisy. A noisy function is one whose output depends on its input 
as well as some internal randomness. Such a tree $\cT$ computes a
function from $S^k$ to the set $L(\cT)$ of leaves of $\cT$ as follows:
on input $\langle{x_1,x_2,\ldots,x_k}\rangle \in S^k$, the computation starts
at the root and determines the next vertex to visit after a vertex $v$
by evaluating $g_v(x_{i_v})$; the leaf reached in the end is the
result of the computation.  If a vertex $i_v = i$ for a vertex $v$, 
then we say that the $i$-th input variable is queried at that vertex.  
We say that the decision tree is {\bf oblivious} if the label $i_v$ of a
vertex $v$ depends only on the level of $v$ (distance from the root).  We say that
an oblivious decision tree is {\bf ordered} if for all $j\in [k]$ all queries to the
the $j$-th input variable appear at consecutive levels. We say that an oblivious decision tree 
is {\bf read-once} if each input variable is queried exactly once.
\end{definition}
  
\begin{remark}
We use the notation $(n,k)$-decision tree to refer to a decision tree for
inputs in $S^k$ where $S = \{0,1\}^n$. 
\end{remark}

\begin{remark}
\label{rem:ordered-read-once}
A read-once decision tree is obviously ordered. Also, an ordered decision tree can be easily made read-once
by collapsing consecutive queries to the same variable into one supernode.
\end{remark}

As in \cite{Goyal05}, in order to capture the noise in a noisy communication network, we define
a special kind of noisy decision tree, {\it Xored-Noise Decision tree} ($\xnd$-tree). 
Here we allow each of the the functions $g_v$ access to its input variable xored with some noise variable.
These noise variables are set according to some distribution based on a noise parameter $\epsilon$, 
but independent of the input.

\begin{definition}[$\xnd$ tree] \hspace{0.2in}
\label{def:xnd-tree}
An $(n,k,D,\epsilon)-\xnd$ tree $\hT$ is an $(n,k)$-noisy decision tree.
It consists of an oblivious decision tree $\cT$ on inputs $S^k$
where $S= \{0,1\}^n \times (\{0,1\}^n)^{|\Lambda|}$ (for some index
set $\Lambda$), and each function $g_v$ has a special form:
\[ g_v(x_{i_v},\zbar_{i_v}) = g'_v(x_{i_v} \oplus z_{{i_v},{\lambda_v}}),\]
for some $g_v': \{0,1\}^n \rightarrow C_v$ and $\lambda_v \in
\Lambda$.  Each input is queried at most $D$ times in the tree.  The
computation of $\hT$ proceeds as follows: on input $x \in
(\{0,1\}^n)^k$, each $z_{i,\lambda} \in \{0,1\}^n$ is chosen
independently according to the binomial distribution $\cB(n,\epsilon)$.
Once the entire input $(\xbar,\zbar) \in S^k$ is determined, we compute $\cT(\xbar,\zbar)$ as in
Definition~\ref{def:decision-tree} above.
\end{definition}

\begin{remark} 
When $k=1$, the trees defined in the above definition correspond to the 
$\gnd$ trees of Goyal, Kindler and Saks~\cite{Goyal05}.
\end{remark}

Let $\cA$ be an algorithm to process inputs from some set $S$. The usefulness of 
$\cA$ to compute some boolean function $f$ on input set $S$ is captured by the notion of its {\it advantage}.

\begin{definition}[Advantage] \hspace{0.2in}
Let $\mu$ be a distribution on some set $S$.  Let $f:S \rightarrow \{+1,-1\}$ 
and $\cA: S \rightarrow C$, where $C$ is some set. Then, the advantage of $\cA$ 
for $f$ under $\mu$ is given by
\[ \adv_{f,\mu}(\cA) = \max_{a:C \rightarrow [-1,+1]} |\E[f(X)a(\cA(X))]|,\]
where $X$ is a random variable taking values in $S$ with distribution $\mu$. 
We will use this notation even when $\cA$ corresponds to a randomized algorithm, 
in which case, the expectation is computed over $X$ as well as the 
internal random choices made by $\cA$.
\end{definition}

\begin{definition}
For a distribution $\mu$ on $\{0,1\}^n$, let 
\[ \alpha_{\mu}(n,D,\epsilon) \defeq \max_{T} \adv_{\oplus,\mu}(T), \] 
where $T$ ranges over all $(n,1,D,\epsilon)-\xnd$ trees. 
\end{definition}
\section{Lower bound for parity}
\label{sec:lb_parity}

Our lower bound proof has two parts. In this section, we will
summarize the results of these two parts of the argument in the form
of lemmas. Then, using these lemmas we will prove the main
theorem. The lemmas themselves will be proved in the next three sections.

\subsection{First part of the proof}
This part of our argument is based on the
observation that in a random planar network, nodes are typically
distributed uniformly over the entire area. By fixing the inputs of
some of the nodes (and thereby making them auxiliary), we can create 
`buffer zones' of auxiliary nodes so that the remaining nodes now fall 
into large number of well-separated large clusters.

\begin{lemma}
\label{lem:network-decomposition}
Suppose $R \leq N^{-\beta}$, for some $\beta > 0$.  Then, with probability 
$1 - o(1)$ over the random variable $\cN(N,R)$, the following holds: if
\begin{quote}
there is a $\delta$-error protocol on $\cN$ with $\epsilon$-noise for computing the
parity function (on $N$ bits) with $T$ transmissions,
\end{quote}
then
\begin{quote}
there is an $(n,k)$-decomposition of $\cN$ and a $\delta$-error $\epsilon$-noise 
$(n,k,d,D)$-protocol with respect to this decomposition for computing parity (on $nk$ bits), where 
$n = \Omega(NR^2)$, $k = \Omega(1/R^2)$, $d = O(T/N)$ and $D = O(TR^2)$.
\end{quote}
\end{lemma}

This lemma is proved in Section~\ref{sec:network-decomposition}.

\subsection{second part of the proof}
In the second part of our argument, we analyze such bounded protocols on
decomposed networks. Our analysis closely follows that of Goyal, Kindler and
Saks~\cite{Goyal05}. For showing lower bounds on the number of transmissions 
in a noisy communication protocol, Goyal et al. translated such protocols into 
$\gnd$ trees.

since we want to analyse bounded protocols for decomposed networks, we first translate
such protocols into $\xnd$-trees. Then we argue that if the inputs come from a product distribution, 
then $\xnd$-trees for computing parity can be rearranged to get ordered $\xnd$-trees, 
and hence read-once noisy decision trees (using Remark~\ref{rem:ordered-read-once}). 

\begin{lemma}[Translation from protocols to read-once decision trees] \hspace{0.2in}
\label{lem:protocol-to-read-once-tree-translation}
For any $\epsilon$-noise $(n,k,d,D)$-protocol $\Pi$ and any distribution $\mu$ on $\{0,1\}^n$, 
there is a read-once noisy $(n,k)$-decision tree $\cT$ such that
\begin{itemize}
\item $\adv_{\oplus,\mu^k}(\cT) \geq \adv_{\oplus,\mu^k}(\Pi)$;
\item $\adv_{\oplus, \mu}(g) \leq \alpha_\mu(n,3D,\epsilon^d)$ for
      every function $g$ that appears in $\cT$.
\end{itemize}
\end{lemma}

Next we observe the following 'product property' for the advantage of 
read-once noisy decision trees.

\begin{lemma}[Advantage of read-once decision trees] \hspace{0.2in}
\label{lem:read-once-tree-advantage}
Let $h: \{0,1\}^n \rightarrow \{+1,-1\}$.  Suppose $\cT$ is a read-once
$(n,k)$-decision tree for computing $f:(\{0, 1\}^n)^k \rightarrow \{+1,-1\}$
defined by $f(\langle{x_1,x_2,\ldots,x_k}\rangle) = \prod_{i=1}^k
h(x_i)$.  Suppose, for each function $\cA$ that appears in $\cT$ we have
$\adv_{h,\mu}(g) \leq \alpha$. Then, $\adv_{f,\mu^k}(\cT) \leq
\alpha^k.$
\end{lemma}

The above two lemmas give the main lemma of the second part of our proof.

\begin{lemma}
\label{lem:bounded-protocol-advantage}
For all distributions $\mu$ on $\{0,1\}^n$ and all $\epsilon$-noise $(n,k,d,D)$-protocol $\Pi$, 
we have 
\[ \adv_{\oplus, \mu^k}(\Pi) \leq \alpha_{\mu}(n,3D,\epsilon^d)^k. \]
\end{lemma}

\begin{proof}
Immediate from Lemma~\ref{lem:protocol-to-read-once-tree-translation} and 
Lemma~\ref{lem:read-once-tree-advantage}.
\end{proof}

Section~\ref{sec:protocol-to-read-once-tree-translation} is devoted to proving
Lemma~\ref{lem:protocol-to-read-once-tree-translation}, and  Section~\ref{sec:read-once-tree-advantage}
proves Lemma~\ref{lem:read-once-tree-advantage}.

\subsection{Putting the two parts together}

To complete the proof of our lower bound, we need the following result of~\cite{Goyal05}.

\begin{definition}
\label{def:sensitivity}
Let $f:\{0,1\}^n \rightarrow \{0,1\}$ be any function. The sensitivity of $f$ at input $x \in \{0,1\}^n$,
denoted $S_x(f)$, is the number of indices $i \in [n]$ such that $f$ changes value upon flipping the $i$th 
bit of $x$. The sensitivity of $f$, denoted $s(f)$, is the maximum of $S_x(f)$ over all $x$.
\end{definition}

\begin{theorem}[Goyal, Kindler and Saks~\cite{Goyal05} (Theorem 32)]
\label{thm:gks-original}
Let $\epsilon \in (0,1/2)$ and $\delta \in (0,1/16)$, and let $f$ be an $n$-variate boolean
function. Any randomized $\gnd$ tree $T$ that for every input $x$, outputs $f(x)$ with probability 
$1-\delta$ when run with noise parameter $\epsilon$ satisfies:
\[ \depth(T) \geq \frac{\epsilon^2 \log(1/4\delta)}{50 \log^2(1/\epsilon)} s(f).\]
\end{theorem}

We will restate the above theorem for the case of parity in terms of advantage of $\xnd$ trees.

\begin{theorem}[Restatement of Theorem~\ref{thm:gks-original}]
\label{thm:gks-restated}
Let $\mu$ be the distribution on $\{0,1\}^n$ defined by $\mu(0^n)=\frac{1}{2}$  
and $\mu(e)=\frac{1}{2n}$ for all $e \in \{0,1\}^n$ of weight 1. Then
\begin{equation}
\alpha_{\mu}(n,D,\epsilon) \leq \max \left(1 - \exp\left(-O\left(\frac{D \log^2 (1/\epsilon)}{\epsilon^2 n}\right)\right),7/8 \right).
\label{eqn:gks-restated}
\end{equation}
\end{theorem}

\begin{proof}[Proof of the restatement.]
Let $\mu$ be as given in the theorem. Theorem~\ref{thm:gks-original} is proved in~\cite{Goyal05} by proving 
an upper bound on the probability that $T$ is correct when $T$ is executed on an input selected at random 
from the distribution $\mu$. Thus any $\gnd$ tree $T$ that makes an average error of at most $\delta < 1/16$ for
computing the parity function $\oplus:\{0,1\}^n \rightarrow \{0,1\}$ on inputs from the distribution $\mu$, 
when run with noise parameter $\epsilon$, must have
\[ \depth(T) \geq \frac{\epsilon^2 \log(1/4\delta)}{50 \log^2(1/\epsilon)} n, \]
since the sensitivity of the parity function $\oplus:\{0,1\}^n \rightarrow \{0,1\}$ is $n$. 
As the RHS of the above equation is strictly decreasing with $\delta$, we conclude that 
any $(n,1,D,\epsilon)-\xnd$ tree $T$ makes an average error of at least $\delta'$ for computing
the parity function on inputs from the distribution $\mu$, where 
\[ \delta' = \min \left(\exp \left(-O \left( \frac{\log^2(1/\epsilon) D}{\epsilon^2 n} \right) \right),1/16 \right). \]
 Thus $\adv_{\oplus,\mu}(T) \leq 1 - 2 \delta'$, which proves the theorem.
\end{proof}

\begin{proof}[Proof of Theorem~\ref{thm:lb-parity}.]
Let $\mu$ be the distribution defined in Theorem~\ref{thm:gks-restated}.
By combining Lemmas~\ref{lem:network-decomposition} and~\ref{lem:bounded-protocol-advantage}, 
we conclude that with probability $1-o(1)$ over the random variable $\cN(N,R)$, the following is true: 
if there is a $\delta$-error protocol on $\cN(N,R)$ with $\epsilon$-noise for computing the
parity function with $T$ transmissions, then
\[ 1-2\delta  \leq \alpha_\mu(n,3D,\epsilon^d)^k,\]
where $n = \Omega(NR^2)$, $k = \Omega(1/R^2)$, $d = O(T/N)$ and $D =O(TR^2)$. 

Since $R \leq N^{-\beta}$, $k = \Omega(1/R^2)$ and $\delta$ is a constant, $\alpha_\mu(n,3D,\epsilon^d)$ 
must be inverse polynomially close to $1$.  Let $k \geq C/R^2$ and $d \leq C'T/N$ for some constants $C,C'$.
From (\ref{eqn:gks-restated}), we thus get
\[ 1 - 2\delta \leq \left(1 - \exp\left(- O\left(\frac{TR^2 \log^2(1/\epsilon^{C'T/N})}{NR^2 \epsilon^{2C'T/N}} 
\right) \right) \right)^{\frac{C}{R^2}}. \]
Denoting $T/N$ by $S$ and simplifying, we have
\[ 1 - 2\delta \leq \exp \left(-\exp \left(-O \left(\frac{S \log^2(1/\epsilon^{C'S})}{\epsilon^{2C'S}} 
\right) \right) \frac{C}{R^2} \right). \]
Taking logarithm and noting that $R \leq N^{-\beta}$,
\[ \exp \left(-O\left(\frac{S \log^2(1/\epsilon^{C'S})}{\epsilon^{2C'S}} \right) \right) \leq 
\frac{N^{-2\beta}}{C} \ln \left(\frac{1}{1-2\delta} \right). \]
From this we get,
\[ \frac{S \log^2(1/\epsilon^{C'S})}{\epsilon^{2C'S}} \geq C'' \log N, \]
for some constant $C''$. This yields $S = \Omega (\log \log N)$ and hence $T = \Omega (N \log \log N)$.
\end{proof}
\section{Decomposition of random planar networks}
\label{sec:network_decomposition}

The random placement of nodes in the unit square typically
arranges them uniformly. We will exploit this uniformity to obtain the
required decomposition.

\begin{lemma}[Chernoff bounds] \hspace{0.2in} 
\label{lm:chernoff}
Let $X$ be the sum of $N$ independent identically distributed indicator random variables. 
Let $\mu=E[X]$. Then, $\Pr[ X \leq \frac{1}{2}\mu]  \leq \exp(-0.15\mu)$. 
\end{lemma}

\begin{proof}
The lemma follows immediately from the following version of the Chernoff bound due to Hoeffding~\cite{Hoeffding63}:
if the random variable $X$ has binomial distribution $\cB(N,p)$, then
\begin{equation}
\label{eq:chernoff}
\Pr[ X \geq (p+\delta)N ] \leq
\left(\frac{p}{p+\delta}\right)^{(p+\delta)N} \hspace{-0.1in} 
\left(\frac{1-p}{1-p-\delta}\right)^{(1-p-\delta)N}.
\end{equation}

To derive the lemma, we consider the random variable $Y=N-X$,
and apply (\ref{eq:chernoff}) with $p=1-\frac{\mu}{N}$ and
$\delta=\frac{\mu}{2N}$, to obtain
\begin{eqnarray*}
\Pr[X \leq \frac{1}{2}\mu] \hspace{-0.1in} &\leq& \hspace{-0.1in} \Pr[Y \geq (p+\delta)N] \\
&\leq& \hspace{-0.1in} \left(1-\frac{\delta}{p+\delta}\right)^{(p+\delta)N} \hspace{-0.1in}
\left(\frac{1-p}{1-p-\delta}\right)^{(1-p-\delta)N} \\
&\leq & \hspace{-0.1in} \exp(-\delta N) \cdot 2^{\frac{\mu}{2}}  \\
&\leq& \hspace{-0.1in} \exp\left( -\frac{1}{2}(1-\ln 2) \mu\right) \\
&\leq& \hspace{-0.1in} \exp(-0.15\mu).
\end{eqnarray*}
\end{proof}

\begin{proof}[Proof of lemma~\ref{lem:network-decomposition}.]
We tessellate the unit square into $M = (\floor{1/R})^2$ cells, each a square
of side $\frac{1}{\floor{1/R}}$. We number the rows and columns of
this tessellation using indices in $\{1,2,\ldots, \floor{1/R}\}$, and
refer to the cell in the $i$-th row and $j$-th column by $c_{ij}$.
The expected number of processors in any one cell is $\mu=N/M$. Since $R
\geq \sqrt{10 \ln N / N}$, we have $\mu \geq 10 \ln N$, and by
Lemma~\ref{lm:chernoff}, the probability that there are fewer than
$\mu/2$ processors in any one cell is is $o(\frac{1}{M})$. So, with
probability $1-o(1)$, all cells have at least $\mu/2 = N/(2M)$
processors.

Now, let $\cS_1 = \{ c_{ij}: i= 1 \pmod{3} \mbox{ and } j= 1
\pmod{3}\}.$ Then, $|\cS_1| \geq M/9$.  For each $c \in \cS_1$, let
the neighborhood of $c$, denoted by $\Gamma(c)$, be the set of (at
most nine) cells that are at distance less than $R$ from $c$. Note that
distinct cells in $\cS_1$ have disjoint neighborhoods. If the total
number of transmissions in the original protocol is $T$, then the average
number of transmissions made from $\Gamma(c)$ as $c$ ranges over $\cS_1$
is at most $9T/M$.  By Markov's inequality, for at least half the cells
$c \in \cS_1$ fewer than $18T/M$ transmissions are made from $\Gamma(c)$.
Let $\cS_2$ be the set of these cells; $|\cS_2| \geq M/18$. For each
cell $c \in \cS_2$, we identify the set $I_c$ of $\ceil{N/(4M)}$
processors that make fewest transmissions.  We are now ready to
describe the decomposition of the planar communication network.

The set of input processors will be $I = \bigcup_{c \in \cS_2}
I_c$. We fix the input of all processors not in $I$ at 0, and treat
them as auxiliary processors. The protocol continues to compute the
parity of the inputs provided to processors in $I$. For $c \in \cS_2$,
let $A_c$ be the set of auxiliary processors in the cells in 
$\Gamma(c)$. Also let $A_0$ be the set of all those auxiliary processors that are not 
in $\Gamma(c)$ for any $c \in \cS_2$. We have thus obtained a decomposition 
$\langle A_0,(I_c,A_c): c \in \cS_2 \rangle$, such that
\begin{enumerate}
\item[(a)] the number of input classes in the decomposition is 
$k=|\cS_2| \geq M/18$;
\item[(b)] each input class has $n=\ceil{\mu/4}$ processors;
\item[(c)] The total number of transmissions made by all processors in 
      $I_c \cup A_c$ is at most $D=18T/M$;
\item[(d)] The total number of transmissions made by any one processor in
  $I_c$ is at most $d= D/n = 72 T/N $. 
\end{enumerate}
Thus we have obtained an $(n,k)$-decomposition of the network $\cN$ and the 
original protocol now reduces to a $\delta$-error $\epsilon$-noise
$(n,k,d,D)$-protocol with respest to this decomposition for computing the parity function on $nk$ bits, 
where $n \geq NR^2/4$, $k \geq \frac{1}{18}\floor{1/R}^2$, $d \leq 72T/N$ and $D \leq 18TR^2$. 
\end{proof}
\section{Translation from protocols to read-once decision trees}
\label{sec:protocol_to_read_once_tree}

In this section, we will first translate bounded protocols for decomposed networks into 
$\xnd$ trees. Then we will show how we can rearrange oblivious decision trees in some cases
to make them ordered. These two steps will then enable us to prove lemma~\ref{lem:protocol-to-read-once-tree-translation}.

\subsection{From bounded protocols to $\xnd$ trees}
\label{sec:protocol-to-xnd}

\begin{lemma}
\label{lem:protocol-to-xnd}
For any $\epsilon$-noise $(n,k,d,D)$-protocol $\Pi$ and any distribution $\mu$ on $(\{0,1\}^n)^k$, 
there is an $(n,k,3D,\epsilon^d)$-$\xnd$ tree $\cT$ such that $\adv_{\oplus,\mu}(\cT)
\geq \adv_{\oplus,\mu}(\Pi)$.
\end{lemma}

\begin{proof}
We will carry out the translation from bounded protocols to $\xnd$ trees via two intermediate models
of communication protocols.

\begin{definition}[Intermediate protocols] \hspace{0.2in}
The following two kinds of protocols are obtained by imposing
restrictions on bounded protocols for decomposed networks of Definition~\ref{def:network-decomposition}.
\begin{description}
\item[Semi-noisy protocol:] An $\epsilon$-noise $(n,k,d,D)$-semi-noisy
protocol differs from an $\epsilon$-noise $(n,k,d,D)$-protocol only in the following
respects.
\begin{enumerate}
\item[(a)] When it is the turn of an input processor to send a message,
it sends only its input bit, whose independent $\epsilon$-noisy copies are
then received by its neighbors.
\item[(b)] A transmission made by an auxiliary processor is not 
subjected to any noise.
\end{enumerate}
\item[Noisy copy protocol:] An $\epsilon$-noise $(n,k,D)$-noisy-copy 
protocol is an $\epsilon$-noise $(n,k,1,D)$-semi-noisy protocol; in other words,
every input processor makes exactly one broadcast of its input bit, so that each 
of its neighbors receives exactly one independent $\epsilon$-noisy copy of this input bit.
\end{description}
\end{definition}

\begin{remark}
In these special kinds of protocols, the messages sent by
the input processors does not depend on the messages these processors
receive.  Thus, we may assume that the input processors make their
transmissions in the beginning of the protocol an appropriate number of
times, and after that the auxiliary processors interact according to
a zero noise protocol.
\end{remark}

\begin{claim}[From  bounded protocol to semi-noisy] 
\label{clm:protocol-to-semi-noisy}
For every function $f:(\{0,1\}^n)^k \rightarrow \{+1,-1\}$,
distribution $\mu$ on $(\{0,1\}^n)^k$ and every $\epsilon$-noise
$(n,k,d,D)$-protocol $\Pi$, there is an $\epsilon$-noise
$(n,k,d,3D)$-semi-noisy protocol $\Pi_1$ such that $\adv_{f,\mu}(\Pi)
\leq \adv_{f,\mu}(\Pi_1)$.
\end{claim}

\begin{claim}[From semi-noisy to noisy-copy]
\label{clm:semi-noisy-to-noisy-copy}
For every function $f:(\{0,1\}^n)^k \rightarrow \{+1,-1\}$,
distribution $\mu$ on $(\{0,1\}^n)^k$ and every $\epsilon$-noise
$(n,k,d,D)$-semi-noisy protocol $\Pi_1$, there is an $\epsilon^d$-noise
$(n,k,D)$-noisy-copy protocol $\Pi_2$ such that
$\adv_{f,\mu}(\Pi_1) \leq \adv_{f,\mu}(\Pi_2)$.
\end{claim}

\begin{claim}[From noisy-copy to $\xnd$ tree]
\label{clm:noisy-copy-to-xnd}
For every function $f:(\{0,1\}^n)^k \rightarrow \{+1,-1\}$,
distribution $\mu$ on $(\{0,1\}^n)^k$ and every $\epsilon$-noise
$(n,k,D)$-noisy-copy protocol $\Pi_2$, there is an
$(n,k,D,\epsilon)$-$\xnd$ tree $\cT$ such that 
$\adv_{f,\mu}(\Pi_2) \leq \adv_{f,\mu}(\cT)$.
\end{claim}

Lemma~\ref{lem:protocol-to-xnd} follows immediately from  
Claims~\ref{clm:protocol-to-semi-noisy},~\ref{clm:semi-noisy-to-noisy-copy} 
and~\ref{clm:noisy-copy-to-xnd}.
\end{proof}

\begin{proof}[Proof of Claim~\ref{clm:protocol-to-semi-noisy}.] 
Fix an $\epsilon$-noise $(n,k,d,D)$-protocol $\Pi$ on a graph $G$. We will construct an
$\epsilon$-noise $(n,k,d,3D)$-semi-noisy protocol $\Pi_1$ on a graph $G_1=(V_1,E_1)$. The graph
$G_1$ will contain $G$ as a subgraph; however, all vertices inherited
from $G$ will correspond to auxiliary processors. In addition, for
each input vertex $v$ of $G$, we will have a new input vertex $v'$ in
$G_1$, which will be connected to $v$ and its neighbors in $G$.
Let $(I = \bigcup_{j=1}^k I_j, A = A_0 \cup \bigcup_{j=1}^k A_j)$ be
the decomposition corresponding to $\Pi$. The decomposition
corresponding to $\Pi_1$ will be $(I' = \bigcup_{j=1}^k I'_j, A' = A_0
\cup \bigcup_{j=1}^k A'_j)$, where $I'_j= \{v': v \in I_j\}$ and $A'_j
= A_j \cup I_j$.  

Suppose $\Pi$ uses $T$ transmissions. For $i=1,2,\ldots,T$ and $v \in V(G)$, let
$b_v[i]$ be the bit received by $v$ when the $i$-th transmission is made;
if $v$ does not receive the $i$-th transmission, we define $b_v[i]$ to be $0$.
The protocol $\Pi_1$ for simulating $\Pi$ will operate in $T$ stages,
one for each transmission made by $\Pi$.  The goal is to ensure that in
the end each auxiliary processor $v$ of $G_1$ constructs a sequence
$b'_v \in \{0,1\}^T$, such that $\langle{b'_v: v \in V(G)} \rangle$ and
$\langle{ b_v : v \in V(G) }\rangle$ (of the protocol $\Pi$) have the
same distribution, for every input in $(\{0,1\}^n)^k$.  This implies
that the outputs of $\Pi'$ and $\Pi$ have the same distribution.
Suppose the first $\ell-1$ stages have been successfully simulated and
$\langle{ b'_v[1,\ldots,\ell-1]: v \in V(G)} \rangle$ have been
appropriately constructed. We now describe how stage $\ell$ is
implemented and $\langle{b'_v[\ell]: v\in V(G) \rangle}$ are
constructed.  If the $\ell$-th transmission in $\Pi$ is made by an
auxiliary processor $v$ in $G$, then it will be simulated in $\Pi_1$
using one noiseless transmission from $v$; if the $\ell$-th transmission is made by an
input vertex $v$ of $G$, then it will be simulated in $\Pi_1$ using two
(noiseless) transmissions from $v$ and one $\epsilon$-noisy transmission
from the corresponding (newly added) input vertex $v'$.

\vspace{0.1in}

{\it $v$ is an auxiliary vertex in $G$}: \hspace{0.2in}
The auxiliary vertex $v$ in $G_1$ operates exactly in
the same fashion as in $G$, and sends a bit $b$, which is received 
without error by all its neighbors. Each neighbor $w \in V(G)$ of $v$
independently sets its bit $b'_w[\ell]$ to be an $\epsilon$-noisy copy of $b$
(using its internal randomness).

\vspace{0.1in}

{\it $v$ is an input vertex in $G$}: \hspace{0.2in} 
The auxiliary vertex $v$ in $G_1$ has all the information that the corresponding input vertex $v$ in $G$
would have had, except the input (which is now given to the new input
vertex $v'$) . So, $v$ transmits (with no noise) two bits, $b_0$ and
$b_1$, corresponding to the two possible input values that $v'$ might
have. Next, the input vertex $v'$ transmits its input $c$; let $c_w$
denote the $\epsilon$-noisy version of $c$ that the neighbor $w \in V(G)$
receives.  Each neighbor $w$ of $v$ now acts as follows: if
$b_0=b_1$, then it sets $b'_w[\ell]$ to be an $\epsilon$-noisy copy of
$b_0$ (using its internal randomness); if $b_0 \neq b_1$, then it 
sets $b'_w[\ell]$ to $b_{c_w}$.
\end{proof}

\begin{proof}[Proof of Claim~\ref{clm:semi-noisy-to-noisy-copy}.]
Let $\Pi_1$ be an $\epsilon$-noise $(n,k,d,D)$-semi-noisy protocol. As remarked
above, all input processors in a semi-noisy protocol can be assumed to
make their transmissions right in the beginning, after which only the
auxiliary processors operate. Thus, each auxiliary processor receives
at most $d$ independent $\epsilon$-noisy copies of the input from each
input processor in its neighborhood. The following lemma of Goyal,
Kindler and Saks~\cite{Goyal05} shows that a processor can generate $d$
independent $\epsilon$-noisy copies of any input from one
$\epsilon^d$-noisy copy. 

\begin{lemma}[Goyal, Kindler and saks~\cite{Goyal05} (Lemma 36)]
Let $t$ be an arbitrary integer, $\epsilon \in (0, 1/2)$ and $\gamma =
\epsilon^t$. There is a randomized algorithm that takes as input a
single bit $b$ and outputs a sequence of $t$ bits and has the property
that if the input is a $\gamma$-noisy copy of $0$ (respectively of
$1$), then the output is a sequence of independent $\epsilon$-noisy
copies of $0$ (respectively of $1$).
\end{lemma}

We  modify the protocol $\Pi_1$ to an $\epsilon^d$-noise $(n,k,D)$-noisy-copy protocol $\Pi_2$ by requiring
that each input processor makes one $\epsilon^d$-noisy transmission of its
input bit. Each auxiliary processor on receiving such a transmission uses its
internal randomness to extract the required $\epsilon$-noisy 
copies. Then onwards the protocol proceeds as before.  
We may now fix internal randomness used by the auxiliary 
processors in such a way that the advantage of the resulting 
protocol for the input distribution $\mu$ 
is at least as good as that of the original protocol. 
Thus, all processors use (deterministic) functions to compute
the bit that they transmit.
\end{proof}

\begin{proof}[Proof of Claim~\ref{clm:noisy-copy-to-xnd}.]
Let $\Pi_2$ be an $\epsilon$-noise $(n,k,D)$-noisy-copy protocol, with the
underlying decomposition $(I = \bigcup_{j=1}^k I_j, A = A_0 \cup
\bigcup_{j=1}^k A_j)$.  We will now show how this protocol can be
simulated using an $(n,k,D,\epsilon)$-$\xnd$ tree $\cT$.  To keep our notation simple, we will
assume (by introducing new edges, if necessary) that (a) all
processors in $A$ are adjacent, and (b) every processor in $A_j$ is
adjacent to every processor in $I_j$. 

Let $T$ be the total number of transmissions in $\Pi_2$.
Let $b_1, b_2, \ldots, b_T$ be the sequence of bits transmitted in $\Pi_2$ by
the auxiliary processors. Suppose, $b_i$ is transmitted by vertex $v \in
I_j$ by computing $g_i(b_1b_2\cdots b_{i-1}, x_j \oplus z_v)$, where
$x_j$ is the the restriction of the input assignment to $I_j$ and
$z_v$ is an $\epsilon$-noisy vector in $\{0,1\}^n$.

The nodes of the $\xnd$ tree $\cT$ are 0-1 sequences of length at most
$T$ (the root is the node at $0$th level and corresponds to the empty sequence). 
The children of the node $b \in \{0,1\}^{i-1}$ $(0 \leq i-1 \leq T-1)$ 
are the two vertices $b0$ and $b1$. Suppose vertex $v \in A_j$ makes the $i$-th
transmission.  The function that $v$ computes to determine what to transmit,
will be used to compute the successor of the nodes at the $i-1$-th
level.  To state this formally, the label of $b \in \{0,1\}^{i-1}$ (at
level $i-1$ in $\cT$) is $(j, h)$, where $h(x_j,z_v) = b \cdot
g_i(b,x_j \oplus z_v)$. (Since our definition requires the function to
return a child of $b$, $h$ returns an extension of $b$ in
$\{0,1\}^i$.)

The set of leaves of $\cT$, $L(\cT)$, is precisely $\{0,1\}^T$. 
Let $a: L(\cT) \rightarrow \{+1,-1\}$ be defined by $a(b_1b_2\cdots b_T) =
(-1)^{b_T}$. Then, it follows from our definitions that
\begin{eqnarray*}
\adv_{\oplus, \mu}(\cT) &\geq& |\E[ \oplus(x) a(\cT(x))]| \\
                        & = &  |\E[ \oplus(x) (-1)^{b_T}]| \\
                        & = &  \adv_{\oplus, \mu} (\Pi_2).
\end{eqnarray*}
\end{proof}

\subsection{Tree rearrangement}
\label{sec:tree-rearrangement}

Our main observation in this section is that oblivious decision trees
can be assumed to be ordered when the inputs come from a product
distribution, and we wish to approximate the parity function. To show
this we will describe a method for rearranging an arbitrary oblivious
decision tree so that it becomes ordered.

\begin{definition}[Tree rearrangement] \hspace{0.2in}
Let $\cT$ and $\cT'$ be oblivious decision trees for the same set of
inputs.  We say that $\cT'$ is a rearrangement of tree $\cT$ if
\begin{itemize}
\item both trees query each variable the same number of times;
\item the functions labelling vertices of $\cT'$ also appear in
  $\cT$ (up to obvious renaming of children);
 formally, for every vertex $\hv$ in $\cT'$ labelled $(i,
  \hg)$, there is a vertex $v$ in $\cT$ labelled $(i,g)$ in $T$ and a
  bijection $\pi:C_{\hv} \rightarrow C_v$ such that 
  $\forall x \in S_i: \hg(x) = \pi(g(x))$.
\end{itemize}
\end{definition}

\begin{lemma}[Ordering lemma] \hspace{0.2in}
\label{lem:tree-ordering}
Let $\mu$ be a product distribution on some set $S^k$. Let $f: S^k \rightarrow \{+1,-1\}$ 
be of the form $f(x_1,x_2,\ldots,x_k) = h(x_1)h(x_2) \cdots h(x_k)$, where $h: S \rightarrow \{+1,-1\}$. 
Then every oblivious decision tree $\cT$ can be rearranged to obtain an ordered oblivious 
decision tree $\hT$ such that $\adv_{f,\mu}(\hT) \geq \adv_{f,\mu}(\cT).$
\end{lemma}

This lemma will follow immediately from the following claim.

\begin{claim}[Move to root] \hspace{0.2in}
\label{clm:move-to-root}
 Let $\mu$ be a product distribution on $S^k$. Let $f: S^k \rightarrow \{+1,-1\}$ be of the form
$f(x_1,x_2,\ldots,x_k) = h(x_1)h(x_2)\cdots h(x_k)$, where $h: S \rightarrow \{+1,-1\}$. 
Let $\cT$ be an oblivious decision tree with inputs in $S^k$
such that the input $x_n$ is queried only at the level just above the
leaves. Then, $\cT$ can be rearranged to obtain a tree $\hT$ where
\begin{enumerate}
\item the input $x_k$ is queried only at the root;
\item for all $j \neq k$, if $x_j$ was queried at level $r$ of $\cT$,
  then $x_j$ is queried at level $r+1$ of $\hT$;
\item $\adv_{f,\mu}(\hT) \geq \adv_{f,\mu}(\cT)$.
\end{enumerate}
\end{claim}

\begin{proof}
Let $X=\langle{X_1,X_2,\ldots,X_k}\rangle$ take values in $S^k$ with
distribution $\mu$; since $\mu$ is a product distribution the $X_i$'s
are independent. Suppose $\cT$ makes $t$ queries to the input.  Let
$\rv_1, \rv_2, \ldots, \rv_{t+1}$ be the random sequence of vertices
visited by the computation of $\cT$ on input $X$. 
Fix $b: L(\cT) \rightarrow [-1,+1]$ such that 
\begin{eqnarray*}
\adv_{f,\mu}(\cT) & = & |\E[h(X_1)h(X_2)\cdots h(X_k) b(\rv_{t+1})]| \\
                 & = & |\E[\E[h(X_1)\cdots  h(X_k) b(g_{\rv_t}(X_k)) \mid \rv_t]]|.
\end{eqnarray*}

Since $X_k$ is queried only at the end, $h(X_1)\ldots h(X_{k-1})$
and $b(g_{\rv_t}(X_k))$ are independent given $\rv_t$, so 
$\E[h(X_1)\ldots h(X_{k-1}) h(X_k) b(g_{\rv_t}(X_k)) \mid \rv_t]$
$=$ $\E[h(X_1)\ldots h(X_{k-1}) \mid \rv_t] \cdot \E[h(X_k) b(g_{\rv_t}(X_k)) \mid \rv_t]$.

Let $\alpha(v)= \E[h(X_1)\ldots h(X_{k-1}) \mid \rv_t=v]$ and
$\beta(v) = \E[h(X_k) b(g_{\rv_t}(X_k)) \mid \rv_t=v]$.  Let $v^* =
\arg \max \beta(v)$; thus, among the functions labelling vertices that
query $X_k$ (at level $t$), $g_{v^*}$ has the best advantage in the tree
for $h$ under the distribution of $X_k$. It is thus natural to expect
(and not hard to verify) that if we replace all queries to $X_k$ by
this query $g_{v^*}$, the overall advantage can only improve. Once
this is done, the last query does not depend on the previous query,
and can, therefore, be moved to the root. We now present the argument
formally.  We have,
\begin{eqnarray}
\label{eqn:best-query}
\adv_{f,\mu}(\cT) & = & |\E[\alpha(\rv_t) \beta(\rv_t)]| \\
                  & \leq & \E[|\alpha(v_T)|] \cdot |\beta(v^*)|. \nonumber
\end{eqnarray}

We are now ready to describe the rearrangement of $\cT$. Let $\cT^-$
be the subtree of $\cT$ consisting of the first $t$ levels of
vertices; thus vertices where $X_k$ is queried in $\cT$ become leaves
in $\cT^-$.  We first make $|C_{v^*}|$ copies 
of $\cT^-$; we refer to these copies by $\cT^{-}_c$ $(c \in C_{v^*})$, 
and assume that the root of $\cT^{-}_c$ is renamed
$c$. In the new tree $\hT$, we have a root with label $\langle{k,g_{v^*}}\rangle$ 
which is connected to the subtrees $\cT^{-}$.  We
claim that $\adv_{f,\mu}(\hT) \geq \adv_{f,\mu}(\cT)$. Indeed, consider
the function $\hb: L(\hT) \rightarrow [-1,+1]$ that takes the value
$\sign(\alpha(v))b(c)$ on the leaf in $\cT^-_c$ corresponding to
$v \in L(\cT^{-})$. Then, we have
\begin{eqnarray}
\label{eqn:new-tree}
\adv_{f,\mu}(\hT) & \geq & |\E[h(X_1)h(X_2)\cdots h(X_k) \hb(\hv_T)]| \\
                 & = & \E[|\alpha(\rv_T)|] \cdot |\beta(v^*)|. \nonumber
\end{eqnarray}

Claim~\ref{clm:move-to-root} now follows by combining (\ref{eqn:best-query}) and (\ref{eqn:new-tree}).
\end{proof}

We are now ready to show how trees computing the parity function
can be reordered, and prove Lemma~\ref{lem:tree-ordering}. The argument
essentially involves repeated application of Claim~\ref{clm:move-to-root}
to place all queries made to a variable in adjacent
levels. We state the argument formally by considering a carefully
defined {\em minimal counterexample}.

\begin{proof}[Proof of Lemma~\ref{lem:tree-ordering}.]
Fix an oblivious decision tree $\cT$. Let the depth $\cT$ be $r$. We
say that there is an {\em alternation at level $\ell \in \{3,\ldots,r\}$}
of $\cT$ if the variable queried at level $\ell$ is queried at a level
before $\ell-1$ but not at level $\ell-1$. Clearly, a tree with no
alternations is an ordered tree. Among all rearrangements of $\cT$,
let $\hT$ be such that
\begin{enumerate}
\item[(P1)] $\adv_{f,\mu}(\hT) \geq \adv_{f,\mu}(\cT)$;
\item[(P2)] among all $\hT$ satisfying (P1), $\hT$ has the
          fewest alternations;
\item[(P3)] among all $\hT$ satisfying (P1) and (P2), 
          the last alternation in $\hT$ is farthest from the root.
\end{enumerate}
We claim that $\hT$ has no alternations.  Let us assume that $\hT$ has
alternations and arrive at a contradiction.  Let $\hT'$ be the tree
obtained from $\hT$ by merging queries on adjacent levels into one
{\em superquery}. That is, if there are $j$ adjacent levels somewhere
in the tree that query $x_i$, with two outcomes, then we replace these
$j$ levels by a single superquery with $2^j$ outcomes.  Note that the
number of alternations in $\hT'$ is the same as in $\hT$. Let $r'$ be
the number of queries in $\hT'$. We consider two cases:

\vspace{0.1in}

{\it $\cT'$ does not have an alternation at level $r'$}: \hspace{0.2in}
Let $x_1$ be the variables queried at level $r'$.  By Claim~\ref{clm:move-to-root}, we
obtain a tree $\hT''$ where the superquery to $x_1$ appears only at the
root, and all other superqueries are shifted one level down.  Now,
however, if each superquery in $\hT''$ is replaced by its
corresponding subtree of queries from $\hT$, then we obtain a
rearrangement of $\hT$ satisfying (P1) and (P2), but with alternation
at a level farther from the root, contradicting (P3).

\vspace{0.1in}

{\it $\cT'$ has an alternation at level $r'$}: \hspace{0.2in}
Suppose $x_1$ is queried at level $r'$, and the previous query to $x_1$ is at level
$r''< r'$ (with no queries to $x_1$ in the levels $r''+1,
r''+2,\ldots, r'-1$). Now, we apply Claim~\ref{clm:move-to-root} to the
subtrees of $\cT'$ rooted at level $r''+1$, thereby obtaining a
rearrangement $\hT''$, where $x_1$ is now queried at levels $r''+1$
instead of at level $r'$.  Clearly, the resulting tree $\hT''$ has
fewer alternations than $\hT'$. Furthermore, if each superquery in
$\hT''$ is replaced by its corresponding tree of queries from $\hT$,
we obtain a rearrangement of $\hT$. It can be verified that this
rearrangement has advantage at least no worse than $\hT$ but has fewer
alternations---contradicting (P2).
\end{proof}

\subsection{Obtaining the read-once decision tree}

\begin{proof}[Proof of Lemma~\ref{lem:protocol-to-read-once-tree-translation}.]
By combining Lemmas~\ref{lem:protocol-to-xnd} and~\ref{lem:tree-ordering}, we
see that $\Pi$ can be converted into an ordered $(n,k,3D,\epsilon^d)$-$\xnd$ tree. 
Since this tree is ordered all queries to any particular variable appear in consecutive levels. 
In our final tree we will combine all these queries into a single query. In 
particular, if there are $\ell \leq 3D$ levels that query $(x_i,z_i)$,
then we collapse them, so as to yield a single query with $2^{\ell}$
outcomes. Note, however, that the result of this query depends not
only on the real input in $x_i \in \{0,1\}^n$ but also on the noise
variable $z_i$. In the final noisy decision tree $\cT$, we regard this
superquery $g(x_i)$ as a noisy function of the input $x_i$, with $z_i$
providing the internal randomness for its computation. Since $g(x_i)$
was derived from an $(n,1,\ell,\epsilon^d)$-$\xnd$ tree with $\ell \leq 3D$ , 
it follows from the definition of $\alpha_{\mu}(n,3D,\epsilon^d)$ that
$\adv_{\oplus,\mu}(g) \leq \alpha_{\mu}(n,3D,\epsilon^d)$.
\end{proof}
\section{Analysis of read-once decision trees}
\label{sec:read_once_tree_advantage}

In this section, we will prove Lemma~\ref{lem:read-once-tree-advantage}.
We will make use of the following proposition.

\begin{proposition}  
\label{prop:adv}
Let $X$ be a random variable taking values in $\{0,1\}^n$ with
distribution $\mu$. Then, for all $f:\{0,1\}^n \rightarrow \{+1,-1\}$, 
$\cA:\{0,1\}^n \rightarrow C$ and $a: C \rightarrow \reals$, 
\[ |\E[f(X)a(\cA(X))]| \leq |a| \cdot \adv_{f,\mu}(\cA), \] 
where $|a| = \max_{c \in C} |a(c)|$.
\end{proposition}

\begin{proof}
\begin{eqnarray*}
|\E[f(X)a(\cA(X))]| & = & |\sum_{c \in C} \E[f(X)a(\cA(X))|\cA(X) = c] \cdot \Pr[\cA(X) = c] | \\
& \leq & \sum_{c \in C} |a(c)| \cdot |\E[f(X)|\cA(X) = c]| \cdot \Pr[\cA(X) = c] \\
& \leq & \max_{c \in C} |a(c)| \cdot \sum_{c \in C} |\E[f(X)|\cA(X) = c]| \cdot \Pr[\cA(X) = c] \\
& = & |a| \cdot \sum_{c \in C} \E[f(X)b(\cA(X))|\cA(X) = c] \cdot \Pr[\cA(X) = c] \\
& \leq & |a| \cdot |\E[f(X)b(\cA(X))]| \\
& \leq & |a| \cdot \adv_{f,\mu}(\cA),
\end{eqnarray*}
where $b: C \rightarrow \{+1,-1\}$ is defined as $b(c) = \sign(\E[f(X)|\cA(X) = c])$ for all $c \in C$.
\end{proof}

\begin{proof}[Proof of Lemma~\ref{lem:read-once-tree-advantage}.] \hspace{0.2in}
Fix $b: L(\cT) \rightarrow [-1,+1]$. 
Let $X$ take values in $(\{0,1\}^n)^k$ with distribution $\mu^k$.
We wish to show that 
\[ |\E[f(X)b(\cT(X))]| \leq \alpha^k.\]

Let the (random) sequence of vertices visited by the computation of
$\cT$ on input $X$ be $\rv_1, \rv_2, \ldots, \rv_k, \rv_{k+1}$. For
$i=1,2,\ldots,k$ and $v$ in level $i$ of the tree (at distance $i-1$
from the root) let
\[\alpha_i(v) = \E[h(X_i)h(X_{i+1})\cdots h(X_{k}) b(\rv_{k+1}) \mid \rv_{i} =v].\]

We will show by reverse induction on $i$ that 
$|\alpha_i(v)| \leq \alpha^{k+1-i}$. The claim will then follow by
taking $i$ to be $1$ and $v$ to be the root of $\cT$. For the base case, we
have
\begin{eqnarray*}
\alpha_k(v) &=&  \E[h(X_{k}) b(\rv_{k+1}) \mid \rv_k = v] \\
            &=&  \E[h(X_k) b(g_v(X_k))] \\
            &\leq& \adv_{h,\mu}(g_v) \leq \alpha.
\end{eqnarray*}

For the induction step assume that $i < k$ and that $|\alpha_{i+1}(w)|
\leq \alpha^{k-i}$ for all vertices $w$ in level $i+1$ of the tree
(at distance $i$ from the root). Then, for a vertex $v$ in level $i$,
we have
\begin{eqnarray*}
|\alpha_i(v)| \hspace{-0.1in} &=& \hspace{-0.1in}
       |\E[h(X_i) h(X_{i+1}) \cdots h(X_{k}) b(\rv_{k+1}) \mid \rv_i = v]| \\
&=& \hspace{-0.1in} |\E[h(X_i) \alpha_{i+1}(g_v(X_i))]| \\
&\leq & \hspace{-0.1in} \adv_{h,\mu}(g_v) \cdot \max_w |\alpha_{i+1}(w)| \\
&\leq& \hspace{-0.1in} \alpha^{k+1-i}.
\end{eqnarray*} 
where we used Proposition~\ref{prop:adv} to justify the second last
inequality, and the induction hypothesis to justify the last
inequality.
\end{proof}
\section{Conclusions}
\label{sec:conclusions}

In this paper, we presented the first lower bound result for the realistic model of wireless
communication networks where there is a restriction on transmission power. Any bit sent by a transmitter
is received (with channel noise) only by receivers which are within the transmission radius of the transmitter.
We showed that to compute the parity of $N$ input bits with constant probability of error, we
need $\Omega(N \log \log N)$ transmissions. This result nicely complements the upper bound result of 
Ying, Srikant and Dullerud~\cite{Ying06}, which showed that $O(N \log \log N)$ transmissions are sufficient 
for computing the sum of all the $N$ bits. Our result also implies that the sum of $N$ bits cannot be approximated 
up to a constant additive error by any constant error protocol for $\cN(N,R)$ using $o(N \log \log N)$ transmissions, 
if $R \leq N^{-\beta}$ for some $\beta > 0$.

Although the techniques of network decomposition and translation of bounded protocols to $\xnd$ trees are fairly general,
some crucial parts of our proof are not. In particular, rearrangement of $\xnd$ trees to get ordered $\xnd$ trees and 
analysis of read-once decision trees used the fact that we are trying to compute the parity function. Thus the same proof
does not yield similar lower bounds for other functions like majority. In subsequent work, we have eliminated
the need for these parts of the proof using entirely different arguments. We have thus succeeded in showing lower bound of 
$\Omega( N \log \log N)$ transmissions for computing the majority and other functions. These results also show that one cannot 
approximate the sum of $N$ bits to within an additive error of $N^\alpha$ (for some $\alpha > 0$) using $o(N \log \log N)$ 
transmissions.

\bibliographystyle{alpha}
\bibliography{parity_full}

\end{document}